\documentclass[aps,twocolumn,showpacs,prl,superscriptaddress,floatfix,longbibliography]{revtex4-1}

\usepackage{amsthm}
\usepackage{amsmath,amssymb,amsmath,amsthm,bm,graphicx,xcolor}
\usepackage{natbib}
\usepackage{color}
\usepackage{tabularx}
\usepackage{epsfig}
\usepackage{amssymb}
\usepackage{graphicx}
\usepackage{wasysym}
\usepackage{dcolumn}
\usepackage{epstopdf}
\usepackage{subfigure}
\usepackage{psfrag}
\usepackage{bbold}

\newtheorem{thm}{Theorem}
\newtheorem{lem}[thm]{Lemma}

\newcommand{\R}{{\mathbb{R}}}
\DeclareMathOperator{\Tr}{Tr}
\DeclareMathOperator{\prob}{Prob}
\DeclareMathOperator{\mean}{\mathbb{E}}

\newcommand{\bra}[1]{\mbox{$\langle #1 |$}}

\newcommand{\ket}[1]{\mbox{$| #1 \rangle$}}

\newcommand{\ketbra}[2]{\left|#1\middle\rangle\!\middle\langle#2\right|}
\newcommand{\norm}[1]{\left\|#1\right\|}

\begin{document}

\title{Reconstructing quantum states from single-party information}

\author{Christian Schilling}
\email{christian.schilling@physics.ox.ac.uk}
\affiliation{Clarendon Laboratory, University of Oxford, Parks Road, Oxford OX1 3PU, United Kingdom}

\author{Carlos L. Benavides-Riveros}
\affiliation{Institut f\"ur Physik, Martin-Luther-Universit\"at Halle-Wittenberg, 06120 Halle, Germany}

\author{P\'eter Vrana}
\affiliation{Department of Geometry, Budapest University of Technology and
Economics, Budapest, Hungary}
\affiliation{QMATH, Department of Mathematical Sciences, University of
Copenhagen, Universitetsparken 5, 2100 Copenhagen, Denmark}

\date{\today}

\begin{abstract}
The possible compatibility of density matrices for single-party subsystems is described by linear constraints on their respective spectra. Whenever some of those quantum marginal constraints are saturated, the total quantum state has a specific, simplified structure. We prove that these remarkable global implications of extremal local information are stable, i.e.~they hold approximately for spectra close to the boundary of the allowed region. Application of this general result to fermionic quantum systems allows us to characterize natural extensions of the Hartree-Fock ansatz and to quantify their accuracy by resorting to one-particle information, only: The fraction of the correlation energy not recovered by such an ansatz can be estimated from above by a simple geometric quantity in the occupation number picture.
\end{abstract}



\maketitle

\paragraph{Introduction.---}
The exact treatment of interacting many-body quantum systems is quite challenging.
This has its origin in the exponential scaling of the total system's Hilbert space dimension as the
system size increases. Consequently, further physical structure needs to be exploited to allow for a
computationally affordable, more efficient (approximate) description: Physical systems typically exhibit only
one- and two-particle interactions. The relevant physical properties are therefore strongly related
to the one- and two-particle picture. In particular the calculation of the ground state amounts to
a variational minimization involving only two-particle reduced density matrices (rather than the full $N$-particle
wave function).
Such approaches exploiting reduced-particle pictures are the most natural and successful ones in practice. Prominent examples for the case of fermions are \emph{Density-Functional-} \cite{Hohenberg} and \emph{Reduced-Density-Matrix-Functional-Theory} \cite{Gilbert,LathRDMFT08} based on the one-particle picture and for the two-particle picture partial solutions of the two-body $N$-representability problem in the form of outer approximation to the set of valid two-particle reduced density matrices (see, e.g., Refs.~\cite{Col63,Percus64,Kummer67,Erdahl78,ErdMazz01,Nakata01,Mazz04,Mazz12,Mazz16}).

A complementary question has gained tremendous physical relevance inspired by the successful development of quantum information theory: How much information in reduced particle descriptions is required to uniquely determine the total multipartite quantum state?  A lot of effort has been spent on this very general and mathematically highly challenging
question. Among several important insights \cite{LindenPRL1,Doisi3party,3out2exp1,EntvsConc,Nirred1,UniqueGrov,Nirred2,WfromBiparty,NonlocCorr,WalterBD,StatefromMeas,CorrinQMP,3out2exp2,GSofFewBody,JoelPRL} there is one result which deserves particular attention: It has been shown \cite{LindenPRL2} that
(generic) pure multipartite quantum states are \emph{uniquely} determined given some specific reduced density matrices of a fraction of 2/3 of the
parties. In this letter we show that in case of so-called \emph{``(quasi)extremal local information''} (specified below) significant structural simplifications for the multipartite state follow already from the \emph{single}-party reduced density matrices. Application of this striking result to the case of fermionic quantum systems allows us to characterize natural extensions of the Hartree-Fock ansatz and quantify the accuracy by resorting to one-fermion information, only.

\paragraph{(Quasi)extremal local information.---}
In the following we consider a multipartite quantum system consisting of
subsystems $\mathcal{S}_1,\mathcal{S}_2,\ldots,\mathcal{S}_N$. We assume for the moment that these \emph{single}-party
systems are distinguishable and that their respective Hilbert spaces
are all finite, $d$-dimensional. We also assume that the total system is in a pure state
$\ket{\Psi} \in \mathcal{H}$ and denote the respective single-party reduced
density-matrices by $\rho_{\mathcal{S}_1},\ldots,\rho_{\mathcal{S}_N}$. The fact that those marginals
originate via partial trace from the same total state $\ket{\Psi}$ exposes strong relations
on them: In a mathematical breakthrough \cite{Kly4,Daft,Kly2,Kly3,Altun} it has been shown that
given density matrices $\rho_{\mathcal{S}_1},\ldots,\rho_{\mathcal{S}_N}$ are compatible if and only if their set $\vec{\lambda}\equiv (\vec{\lambda}_{\mathcal{S}_1},\ldots,\vec{\lambda}_{\mathcal{S}_N})$
of spectra obeys specific linear constraints (with integer coefficients). By arranging each spectrum in decreasing order those
\emph{quantum marginal constraints} define a high-dimensional polytope $\mathcal{P}\subset \mathbb{R}^{Nd}$ of
mathematically possible spectra $\vec{\lambda}$ (see Ref.~\cite{CSQMath12} for an elementary review). The form of this polytope strongly depends on the dimension of the local Hilbert spaces, the number of subsystems and on possible additional restrictions (as, e.g., symmetries of $\ket{\Psi}$).
The prime example is the one of $N$ qubits: Their \emph{single} qubit-reduced density matrices $\rho_1,\ldots,\rho_N$  are compatible
if and only if their spectra obey the inequalities \cite{Higushi}
\begin{equation}\label{eq:Hig}
D_i(\vec{\lambda})\equiv -\lambda_i^{(2)}+ \sum_{j\neq i} \lambda_j^{(2)} \geq 0\,,
\end{equation}
for all $i$. Here, $\lambda_i^{(2)}$ denotes the smaller eigenvalue of the $i$-th qubit reduced density matrix $\rho_i$, $\lambda_i^{(1)}=1-\lambda_i^{(2)}$, and we introduce
$\rho_i \equiv \lambda_i^{(1)} \ket{1}\!\bra{1}+\lambda_i^{(2)} \ket{2}\!\bra{2}$.

Particular physical relevance of the quantum marginal constraints is given whenever the total spectral vector $\vec{\lambda}$ lies on the boundary $\partial \mathcal{P}$ of the polytope $\mathcal{P}$. In case of such \emph{``extremal local information''} the corresponding physical system is not only limited in its response to external unitary perturbations but remarkable structural simplifications follow for the quantum state of the total system. To demonstrate this, saturation of a constraint $D_i$ (\ref{eq:Hig}) implies a drastic reduction to only $N$ of the $2^N$ product states $\ket{i_1,\ldots,i_N}$, $i_k=1,2$, which contribute to the total state $\ket{\Psi}$. For $i=1$, e.g., these are $\ket{1,\ldots,1}$ and $\ket{2,2,1,1,\ldots}$, $\ket{2,1,2,1,1,\ldots}$,\ldots,$\ket{2,1,\ldots,1,2}$. In general, such \emph{selection rules} implied by extremal local information can be stated in a compact form: Whenever a constraint $D\big(\vec{\lambda}_{\mathcal{S}_1},\ldots,\vec{\lambda}_{\mathcal{S}_N}\big)\geq 0$ is saturated, any compatible total state $\ket{\Psi}$ fulfills \cite{Alex,GLS}
\begin{equation}\label{eq:Dhat}
\hat{D}_{\Psi}\, \ket{\Psi} \equiv D\Big[\big(\hat{n}_{\Psi,\mathcal{S}_1}^{(i)}\big)_{i=1}^d,\ldots,\big(\hat{n}_{\Psi,\mathcal{S}_N}^{(i)}\big)_{i=1}^d\Big]\,\ket{\Psi}=0\,,
\end{equation}
where, e.g., $\hat{n}_{\Psi,\mathcal{S}_1}^{(i)}\equiv \ket{i}_{\!\mathcal{S}_1 \mathcal{S}_1}\!\bra{i}\otimes \mathbb{1}_{\mathcal{S}_2}\otimes \ldots \otimes \mathbb{1}_{\mathcal{S}_N}$ with $\ket{i}_{\mathcal{S}_1}$ the eigenvector of $\rho_{\mathcal{S}_1}$ corresponding to its $i$-th largest eigenvalue $\lambda_{\mathcal{S}_1}^{(i)}$.
Due to the linearity of the quantum marginal constraints, one has $\bra{\Psi}\hat{D}_\Psi\ket{\Psi}=D\big(\vec{\lambda}_{\mathcal{S}_1},\ldots,\vec{\lambda}_{\mathcal{S}_N}\big)$. Hence, Eq.~(\ref{eq:Dhat}) amounts to a simple selection rule for the product states built up from the local eigenstates $\ket{i_1}_{\mathcal{S}_1},\ldots,\ket{i_N}_{\mathcal{S}_N}$: The general expansion
\begin{equation}\label{eq:Psi}
\ket{\Psi} =  \sum_{(i_1,\ldots,i_N)\in \mathcal{I}_D} c_{i_1,\ldots,i_N}\,\ket{i_1,\ldots,i_N}\,,
\end{equation}
restricts to configurations $(i_1,\ldots,i_N)\in \mathcal{I}_D$, namely those fulfilling
$\hat{D}_\Psi \ket{i_1,\ldots,i_N}=0$ (see example above). This remarkable structural simplification (\ref{eq:Psi}) based on (\ref{eq:Dhat}) has a deep mathematical origin \cite{Alex}. On the other hand, however, it concerns a set of quantum states of measure zero. The crucial question from a physical viewpoint is therefore the one about the stability of conditions (\ref{eq:Dhat}) and (\ref{eq:Psi}), respectively: Does $\vec{\lambda}$ close to the polytope boundary (\emph{``quasiextremal local information''}), $D\big(\vec{\lambda}_{\mathcal{S}_1},\ldots,\vec{\lambda}_{\mathcal{S}_N}\big)\approx 0$, imply approximately that simplified structure (\ref{eq:Psi})? In the following we answer this question by `Yes'.

Due to the particular relevance of the quantum marginal constraints for fermions and due to the elegance of the second quantization we present our derivation for the setting of $N$ identical fermions with an underlying one-particle Hilbert space $\mathcal{H}_1^{(d)}$ of dimension $d$. The case of a quantum system without exchange symmetry is treated in the same fashion. Let us now consider one of the quantum marginal constraints, called \emph{generalized Pauli constraints},
\begin{equation}\label{eq:gpc}
D(\vec{\lambda}) = \kappa_0 +\sum_{j=1}^d \kappa_j \lambda^{(j)} \geq 0\,,
\end{equation}
with $\kappa_i\in\mathbb{Z}$. For a fixed state $\ket{\Psi}\in \wedge^N[\mathcal{H}_1^{(d)}]$, with one-particle reduced density matrix,
\begin{equation}\label{eq:1rdm}
\rho_1\equiv N \mbox{Tr}_{N-1}\big[\ket{\Psi}\!\bra{\Psi}\big]\equiv \sum_{j=1}^d \lambda^{(j)} \ket{j}\!\bra{j}\,,
\end{equation}
we define the corresponding $\hat{D}_{\Psi}$-operator (\ref{eq:Dhat}). By introducing the particle number operators $\hat{n}_\Psi^{(j)}$ in second quantization with respect to the eigenvectors $\ket{j}$ of $\rho_1$ it reads
\begin{equation}\label{eq:Dhatf}
\hat{D}_\Psi=\kappa_0 \mathbb{1}+\sum_{i=1}^d\kappa_j \hat{n}_\Psi^{(j)}\,.
\end{equation}
The general idea is now to define an artificial time-evolution/flow acting on $\ket{\Psi}\equiv \ket{\Psi(t=0)}$ with the effect that the respective vector $\vec{\lambda}(t)$ of $\ket{\Psi(t)}$ converges to the polytope facet $F_D$ (defined by $D=0$). If the change of $\ket{\Psi(t)}$ is not too large we can then relate $\ket{\Psi}$ to the state $\ket{\Psi_\infty}\equiv \lim_{t\rightarrow \infty}\ket{\Psi(t)}$ which has the simplified structure implied by Eq.~(\ref{eq:Dhat}).

A promising candidate for such a flow is defined by the differential equation (justified retrospectively),
\begin{equation}\label{eq:flow}
\frac{d}{dt}\ket{\Psi(t)}=-\big(\mathbb{1}-\ketbra{\Psi(t)}{\Psi(t)}\big)\,\hat{D}_{\Psi(t)}\ket{\Psi(t)}.
\end{equation}
This differential equation with the initial condition $\ket{\Psi(0)}=\ket{\Psi}$ has a unique solution as long as $\vec{\lambda}(t)$ remains nondegenerate. Note that the first factor guarantees that the $L^2$-norm $\norm{\Psi(t)}\equiv \sqrt{\bra{\Psi(t)}\Psi(t)\rangle}$ is constant. The one-particle reduced density matrix evolves as
\begin{equation}\label{eq:rhochange}
\dot{\rho}_1(t)=N \mbox{Tr}_{N-1}\big[\ket{\dot{\Psi}(t)}\!\bra{\Psi(t)}+\ket{\Psi(t)}\!\bra{\dot{\Psi}(t)}\big]\,,
\end{equation}
where the `dot' stands for $\frac{d}{dt}$. By perturbation theory, we can determine the change $\frac{d}{dt}\lambda^{(i)}(t)$ of all eigenvalues and therefore the change of $D(\vec{\lambda}(t))$, as well. An elementary but lengthy calculation yields \footnote{See the Supplemental Material at url for technical details used in the flow-based proof of the structural implications of quasiextremal local information, a brute-force proof for the setting $(N,d)=(3,6)$ including the case of degenerate occupation numbers and technical details used for relating the energy and the $N$-fermion picture.}
\begin{equation}\label{eq:Dchange}
\frac{d}{dt}D(\vec{\lambda}(t)) = -2\,\mbox{Var}_{\Psi(t)}\!\hat{D}_{\Psi(t)}\,.
\end{equation}
Eq.~(\ref{eq:Dchange}) justifies retrospectively the definition of the flow (\ref{eq:flow}). It reduces the distance $D(\vec{\lambda}(t))$ of $\vec{\lambda}(t)$ to the polytope facet $F_D$ as long as the variance $\mbox{Var}_{\Psi(t)}\!\hat{D}_{\Psi(t)} \equiv \langle\hat{D}_{\Psi(t)}^{\,2}\rangle_{\Psi(t)}-\langle\hat{D}_{\Psi(t)}\rangle_{\Psi(t)}^{\,2}$ does not vanish. Even further, since $\hat{D}_{\Psi(t)}$ has an integer-valued spectrum we can conclude that whenever $D(\vec{\lambda})$ is small but nonzero $\ket{\Psi(t)}$ has weight in more than one eigenspace of $\hat{D}_{\Psi(t)}$, i.e.~the variance cannot vanish. To be more specific, one proves \cite{Note1}
\begin{equation}\label{eq:Varest}
\mbox{Var}_{\Psi(t)}\!\hat{D}_{\Psi(t)} \geq D(\vec{\lambda}(t)) \, \big[1-D(\vec{\lambda}(t))\big]\,.
\end{equation}
which together with (\ref{eq:Dchange}) leads to (assuming $D(\vec{\lambda}(t))\leq \frac{1}{2}$)
\begin{equation}\label{eq:DchvsD}
\frac{d}{dt}D(\vec{\lambda}(t))\le - 2 D(\vec{\lambda}(t)) \, \big[1-D(\vec{\lambda}(t))\big]\leq -D(\vec{\lambda}(t))\,.
\end{equation}
Eq.~(\ref{eq:DchvsD}) implies an exponential decay, $0\leq D(\vec{\lambda}(t))\le D(\vec{\lambda}(0))e^{-t}$. Hence, $D(\vec{\lambda}(t))\rightarrow 0$ for $t\rightarrow \infty$, as we were hoping for.

It is important that the flow (\ref{eq:flow}) does not change the quantum state $\ket{\Psi}$ too much. To confirm this we observe
\begin{equation}\label{eq:NPsich}
\begin{split}
\big\|\dot{\Psi}(t)\big\|^2
 & = \big\|(\mathbb{1}-\ketbra{\Psi(t)}{\Psi(t)})\hat{D}_{\Psi(t)}\Psi(t)\big\|^2  \\
 & = \mbox{Var}_{\Psi(t)}\hat{D}_{\Psi(t)}\,.
\end{split}
\end{equation}
This allows us by integrating both sides of Eq.~(\ref{eq:NPsich}) between $0\le t_1\le t_2$, to estimate the change of the quantum state, leading to \cite{Note1}
\begin{equation}\label{eq:Psich}
\norm{\Psi(t_2)-\Psi(t_1)} \le \sqrt{2D(\vec{\lambda}(t_1))}-\sqrt{2D(\vec{\lambda}(t_2))}\,.
\end{equation}
Eq.~(\ref{eq:Psich}) implies that $\Psi_\infty\equiv\lim_{t\to\infty}\Psi(t)$ exists, and
\begin{equation}\label{eq:Psidiff}
\norm{\Psi_\infty-\Psi}\le \sqrt{2D(\vec{\lambda})}\,.
\end{equation}
Eq.~(\ref{eq:Psidiff}) is the result we were aiming at. A quantum state $\ket{\Psi}$ is close to a quantum state $\ket{\Psi_\infty}$ exhibiting the simplified structure \ref{eq:Psi} implied by Eq.~(\ref{eq:Dhat}) whenever its occupation number vector $\vec{\lambda}$ is close to the polytope facet $F_D$.

Note that in the above derivation we have assumed that the eigenvalues $\lambda^{(i)}(t)$ are nondegenerate for all times $t$. This can be ensured if the initial $D(\vec{\lambda})$ is small compared to $\min_i\{\lambda^{(i)}-\lambda^{(i+1)}\}$, since the latter one depends continuously on $\ket{\Psi}$ and the former puts a uniform (in $t$) upper bound on $\norm{\Psi(t)-\Psi}$ (recall Eq.~(\ref{eq:Psich})). There is little doubt that our results also hold in the case of (quasi-)degenerate occupation numbers: In \cite{Note1} we present a brute-force proof for the specific setting $(N,d)=(3,6)$ including the case of degenerate occupation numbers. Monte-Carlo sampling allows us to extend this to the settings $(3,7),(4,7)$.

We also would like to stress that according to Eq.~(\ref{eq:Psidiff}) $\ket{\Psi}$  has the approximately simplified structure (\ref{eq:Psi})
with respect to the natural orbitals of $\ket{\Psi_{\infty}}$.
Hence, the ultimate result states that any multipartite quantum state $\ket{\Psi}$ can be approximated by the structural simplified form corresponding to saturation of a quantum marginal constraint $D$ up to an error bounded by
\begin{equation}\label{eq:WEst}
1-\|\hat{P}_D^{(\mathcal{B}_1)}\Psi\|^2\leq 2 D(\vec{\lambda})\,,
\end{equation}
namely the distance of $\vec{\lambda}$ to the corresponding polytope facet $F_D$.
Here, $\mathcal{B}_1$ denotes the unspecified reference basis (bases) for the local Hilbert space(s) $\mathcal{H}_1^{(d)}$ and $\hat{P}_D^{(\mathcal{B}_1)}$ denotes the projection operator on the zero-eigenspace of the corresponding $\hat{D}_{\mathcal{B}_1}$-operator (\ref{eq:Dhatf}) (or (\ref{eq:Dhat}) for distinguishable subsystems). Only for the case of no approximate degeneracies, $\mathcal{B}_1$ is given by the eigenvectors of the single-party marginal(s).

\paragraph{Extensions of the Hartree-Fock ansatz.---}
In the following, we discuss an interesting implication of our main result (\ref{eq:WEst}) concerning natural extensions of the Hartree-Fock variational ansatz. The Hartree-Fock ansatz restricts the variational ground state search to the manifold $\mathcal{M}_{0,0}$ of single Slater determinants.
The index $(0,0)$ of $\mathcal{M}_{0,0}$ should indicate that the corresponding active space consists of $0$ active electrons and $0$ active orbitals. Clearly, $\ket{\Psi} \in \mathcal{M}_{0,0}$ is equivalent to $\vec{\lambda}=(1,\ldots,1,0,\ldots)\equiv \vec{\lambda}_{HF}$. This allows us to characterize the Hartree-Fock ansatz alternatively: It corresponds to $\{\vec{\lambda}_{HF}\}$, a (zero-dimensional) facet of the polytope $\mathcal{P}$ defined by all those $\vec{\lambda}$ whose entries $\lambda^{(i)}$ attain one of the bounds of the Pauli constraints $0\leq \lambda^{(i)}\leq 1$ for all $i$. In case of increasing correlations the size of the active space needs to be increased. In general, if $r\equiv N-N_a$ electrons are frozen and additional $s\equiv d-d_a-N+N_a$ orbitals are inactive the corresponding quantum states form an active space manifold $\mathcal{M}_{N_a,d_a}$. $\mathcal{M}_{N_a,d_a}$ can alternatively be characterized by the corresponding facet of $\mathcal{P}$ formed by the vectors $\vec{\lambda}=(\underbrace{1,\ldots,1}_{r},\vec{\lambda}_a,\underbrace{0,\ldots,0}_s)$ saturating the collective Pauli constraint $S_{r,s}(\vec{\lambda})\equiv {\sum_{i=1}^r (1-\lambda^{(i)})}+\sum_{j=d-s+1}^d \lambda^{(j)} \geq 0$.
The respective variational method based on $\mathcal{M}_{N_a,d_a}\equiv \mathcal{M}_{S_{r,s}}$ is called
\emph{Complete Active Space Self-Consistent-Field} (CASSCF) ansatz \cite{QChemJensen}. Its relevance and success for the analysis of ground states of atoms and molecules is based on the fact that those systems approximately saturate some of the Pauli constraints.
In a similar fashion we can also define extremal parts of the polytope $\mathcal{P}$ by saturation of one (or more) \emph{generalized} Pauli constraint(s) $D$. This would then give rise to a respective state manifold $\mathcal{M}_D$ containing all states which have the simplified structure (\ref{eq:Psi}) implied by (\ref{eq:Dhat}) with respect to some local reference basis $\mathcal{B}_1$.
The debate \cite{Kly1,CS2013,BenavLiQuasi,Mazz14,CSthesis,MazzOpen,CSQuasipinning,Mazzagain,BenavQuasi2,RDMFT,Alex,CS2015Hubbard,RDMFT2,CS2016a,NewMazziotti,CS2016b,CSQ}
on the physical relevance of the generalized Pauli constraints is, however, still ongoing.

Based on our main result (\ref{eq:WEst}) we will now provide an intriguing estimate of the numerical quality of such variational ansatzes defined through extremal one-particle information. By quantitative means, we will show that such an ansatz based on a facet $F_D$ of $\mathcal{P}$ corresponding to saturation of some (generalized) Pauli constraint(s) reconstructs most of the correlation energy whenever $\vec{\lambda}_0$ of the ground state $\ket{\Psi_0}$ of the Hamiltonian $\hat{H}$ lies close to $F_D$. For this, we denote by $\ket{\Psi_D}$ the variational ground state with corresponding vector $\vec{\lambda}_D$ and energy $E_D$. Moreover, we introduce the Hartree-Fock ground state $\ket{\Psi_{HF}}$ which has the energy $E_{HF}$ (see also Fig.~\ref{fig:PolyEn}).

\begin{figure}[b]
\centering
\includegraphics[width=4.5cm]{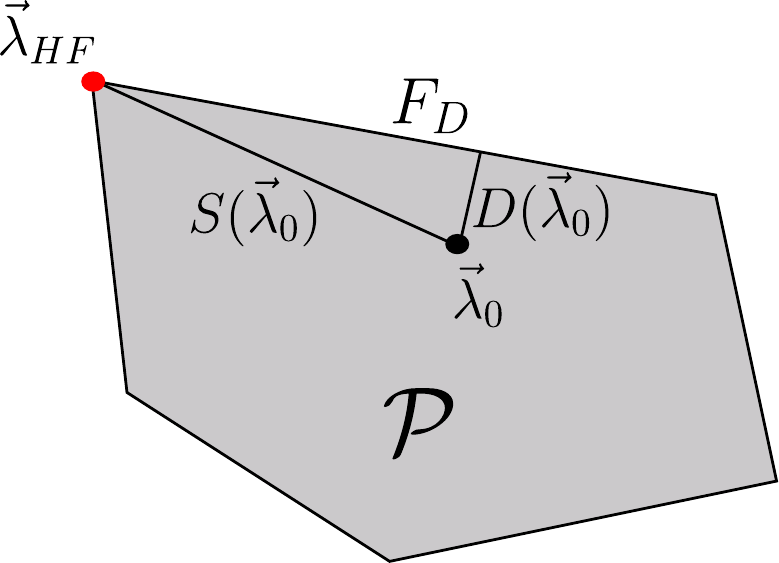}
\hspace{0.5cm}
\includegraphics[width=2.8cm]{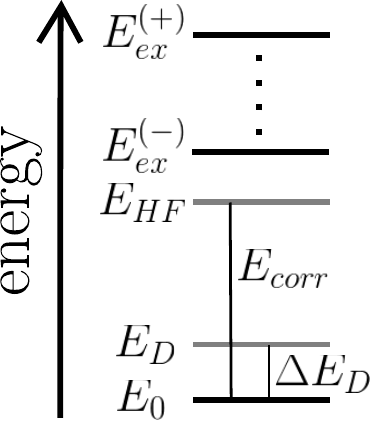}
\caption{Left: Illustration of the distances $D(\vec{\lambda}_0)$ and $S(\vec{\lambda}_0)$ of $\vec{\lambda}_0$ (\ref{eq:1rdm}) for the ground state $\ket{\Psi_0}$ to the polytope facet $F_D$ (defining the variational ansatz) and to the Hartree-Fock point $\vec{\lambda}_{HF}$, respectively.
Right: Energy spectrum of the Hamiltonian $\hat{H}$ is shown in black, Hartree-Fock energy $E_{HF}$ and the variational energy $E_D$ corresponding to $F_D$ in gray.}
\label{fig:PolyEn}
\end{figure}

First, by referring to the energy gap $E_{ex}^{(-)}-E_0$ between the ground state and first excited state we can relate for any quantum state $\ket{\tilde{\Psi}}$ its energy expectation value $\tilde{E}\equiv \bra{\tilde{\Psi}}\hat{H} \ket{\tilde{\Psi}}$ relative to the ground state energy to its $L^2$-weight outside the ground state subspace $\mathcal{H}_{E_0}$ (see Lemma \ref{lem:EvsN} in \cite{Note1}). Indeed, when $\tilde{E}$ lies close to $E_0$ relative to the gap most of the weight of $\ket{\tilde{\Psi}}$ has to lie within the ground state subspace. On the other hand, if $\tilde{E}-E_0>0$ at least some weight has to lie on the excited states of $\hat{H}$.
Combining this relation between the energy and the $N$-particle picture
with the one between the one- and the $N$-particle picture derived above (cf.~Eq.~(\ref{eq:WEst})) leads to a striking relation between the energy and the one-particle picture (assuming $\hat{H}$ has a unique ground state): The energy error $\Delta E_D \equiv E_D-E_0$ of the variational ansatz based on the polytope-facet $F_D$ is bounded from above \cite{Note1}
\begin{equation}\label{eq:estimate1}
 \Delta E_D \leq C\, D(\vec{\lambda}_0)\,,
\end{equation}
where $C\equiv2(E_{ex}^{(+)}-E_0)$. From a geometrical viewpoint, Eq.~(\ref{eq:estimate1}) states that the energy error is linearly bounded from above by the distance $D(\vec{\lambda}_0)$ of the spectral vector $\vec{\lambda}_0$ of the ground state to the respective facet $F_D$.

Even a more sophisticated estimate can be found \cite{Note1}
\begin{equation}\label{eq:estimate2}
\frac{\Delta E_D}{E_{corr}}
 \leq K\, \frac{D(\vec{\lambda}_0)}{S(\vec{\lambda}_0)}\,,
\end{equation}
where $E_{corr}\equiv E_{HF}-E_0$ is the omnipresent correlation energy and $K\equiv 2 N (E_{ex}^{(+)}-E_0)/(E_{ex}^{(-)}-E_0)$.
Estimate (\ref{eq:estimate2}) states that the fraction $\Delta E/E_{\rm corr}$ of the total correlation energy $E_{corr}$ not recovered by the variational ansatz based on $F_D$ is bounded by the ratio $D(\vec{\lambda}_0)/S(\vec{\lambda}_0)$ of $\vec{\lambda}_0$'s distances $D(\vec{\lambda}_0)$ to $F_D$ and $S(\vec{\lambda}_0)$ to the Hartree-Fock point $\vec{\lambda}_{HF}$ (see also Fig.~\ref{fig:PolyEn}).
We also would like to stress that the universality of (\ref{eq:estimate1}) and (\ref{eq:estimate2}), holding for all $\hat{H}$ with non-degenerate ground state, inevitably implies that the prefactors $C$ and $K$ depend linearly on $E_{ex}^{(+)}$. This is due to the fact that without further specification of the Hamiltonian the small weight of $\ket{\Psi_D}$ outside the ground state space could, at least in principle, lie on the highest excited state. In practice, however, one can expect that this weight lies mainly on the lowest few excited states which would improve the constants $C$ and $K$, significantly.

\paragraph{Summary and Conclusion.---}
We have proven by exploiting an elegant flow-approach that significant structural simplifications follow for the total, multipartite quantum state whenever the spectra of the \emph{single}-party marginals lie close to (or even on) the boundary of the allowed region (polytope).
This implication of \emph{``quasiextremal local information''} is remarkable since the unique determination of \emph{generic} quantum states of, e.g., a 300-party system requires the knowledge of marginals of size possibly up to 200 \cite{LindenPRL2} (which is rather difficult to access in experiments).

A comment is in order regarding the possible presence of quasiextremal local information. In contrast to \emph{generic} quantum states, typically not exhibiting quasiextremal local information, the situation can be quite different for \emph{ground states} of \emph{physical systems} with \emph{local interactions}. This is for instance the case for systems of confined fermions. There, one observes an (approximate) saturation of some Pauli constraints which has its origin in the strong conflict between energy minimization and exchange symmetry. This is also the reason for the success of the \emph{Complete-Active Space Self-Consistent Field} method exploiting the corresponding structural simplifications (`reduced active space') based on Eq.~(\ref{eq:Dhat}).

In the form of estimates (\ref{eq:estimate1}) and (\ref{eq:estimate2}) we have also revealed an intriguing universal relation between the numerical quality of prominent variational methods and the distance of the one-fermion density matrix $\rho_1$ to the polytope boundary. We expect that this may stimulate fruitful follow-up ideas. Just to name one, recall that \emph{Reduced Density Matrix Functional Theory} seeks an exact functional $\mathcal{F}$ of $\rho_1$ whose minimization yields the correct ground state energy and the respective  $\rho_1$.
While the existence of the polytope $\mathcal{P}$ is taken into account so far \emph{only} by the restriction of the minimization process to $\mathcal{P}$ \cite{RDMFT}, our work shows that the \emph{whole vicinity} of the polytope boundary $\partial \mathcal{P}$ should play an important role: Density matrices $\rho_1$ near $\partial \mathcal{P}$ correspond to very specific $N$-fermion quantum states. Hence, the exact functional $\mathcal{F}$  needs to include a term with a strongly repulsive behavior close to the polytope-boundary.

\begin{acknowledgements}
We are grateful to D.\hspace{0.5mm}Gross and A.\hspace{0.5mm}Lopes for various stimulating and illuminating discussions.
In addition, we thank M.\hspace{0.5mm}Christandl, D.\hspace{0.5mm}Jaksch, V.\hspace{0.5mm}Vedral, M.\hspace{0.5mm}Walter and Z.\hspace{0.5mm}Zimbor\'as for helpful discussions. We acknowledge financial support from the Swiss National Science Foundation (Grant P2EZP2 152190), the Oxford Martin Programme on Bio-Inspired Quantum Technologies and the UK Engineering and Physical Sciences Research Council (Grant EP/P007155/1) (C.S.), the European
Research Council (ERC Grant Agreement no 337603), the Danish Council for Independent Research (Sapere Aude) and VILLUM FONDEN via the QMATH
Centre of Excellence (Grant No. 10059) (P.V.).
\end{acknowledgements}

\bibliography{EvsNvs1}

\onecolumngrid
\newpage
\begin{center}\Large{\textbf{Supplemental Material}}
\end{center}
\setcounter{equation}{0}
\setcounter{figure}{0}
\setcounter{table}{0}
\makeatletter
\renewcommand{\theequation}{S\arabic{equation}}
\renewcommand{\thefigure}{S\arabic{figure}}
\vspace{0.5cm}

\section{Flow-based proof of the stability}\label{app:flow}
We provide technical details used in the flow-based derivation of the structural
implications in case of quasiextremal local information.
\subsection{Proof of Eq.~\eqref{eq:Dchange}}
We recall Eqs.~(\ref{eq:gpc}-\ref{eq:rhochange}), in particular the form of the one-particle reduced density matrix $\rho_1(t)\equiv \sum_{i=1}^d \lambda^{(i)}(t) \ket{i(t)}\!\bra{i(t)}$, and introduce the particle number operator $\hat{n}_i(t)$ acting on the $N$-fermion Hilbert space for the eigenstate $\ket{i(t)}$ of $\rho_1(t)$. Rayleigh-Schr\"odinger perturbation theory applied to $\rho_1(t)$ then yields (where the `dot' stands for $\frac{d}{dt}$)
\begin{equation}
\begin{split}
\frac{d}{dt}D(\vec{\lambda}(t))
 & = \sum_{i=1}^d\kappa_i\dot{\lambda}^{(i)}(t)  \\
 & = \sum_{i=1}^d \kappa_i \bra{i(t)} \dot{\rho}_1(t) \ket{i(t)} \\
 & = \sum_{i=1}^d \kappa_i \Tr_1\!\big[\ket{i(t)}\!\bra{i(t)}\dot{\rho}_1(t)\big] \\
 & = N \sum_{i=1}^d \kappa_i \Tr_1\!\Big[\ket{i(t)}\!\bra{i(t)}\Tr_{N-1}\!\big[\ket{\dot{\Psi}(t)}\!\bra{\Psi(t)}+\ket{\Psi(t)}\!\bra{\dot{\Psi}(t)}\big]       \Big] \\
 & = \sum_{i=1}^d\Tr_N\!\left[\kappa_i \hat{n}_i(t)\left(\ket{\dot{\Psi}(t)}\!\bra{\Psi(t)}+\ket{\Psi(t)}\!\bra{\dot{\Psi}(t)}\right)\right]  \\
 & = \Tr_N\!\left[\hat{D}_{\Psi(t)}\left(\ket{\dot{\Psi}(t)}\!\bra{\Psi(t)}+\ket{\Psi(t)}\!\bra{\dot{\Psi}(t)}\right)\right]  \\
 & = -\Tr_N\!\Big[\hat{D}_{\Psi(t)}\Big(\big(\mathbb{1}-\ketbra{\Psi(t)}{\Psi(t)}\big)\hat{D}_{\Psi(t)}\ketbra{\Psi(t)}{\Psi(t)}
 +\ketbra{\Psi(t)}{\Psi(t)}\hat{D}_{\Psi(t)}\big(\mathbb{1}-\ketbra{\Psi(t)}{\Psi(t)}\big)\Big)\Big]  \\
 & = -2\Big\{\Tr_N\!\big[\hat{D}_{\Psi(t)}^2\ketbra{\Psi(t)}{\Psi(t)}\big]-\left(\Tr_N\!\big[ \hat{D}_{\Psi(t)}\ketbra{\Psi(t)}{\Psi(t)}\big]\right)^2\Big\} \\
 & = -2\,\mbox{Var}_{\Psi(t)}\hat{D}_{\Psi(t)}\,.
\end{split}
\end{equation}
The formula for the first order correction according to Rayleigh-Schr\"odinger perturbation theory has been used in the second line. In the fourth to last line we have used in particular $\Tr_N\!\big[\ket{\Psi(t)}\!\bra{\dot{\Psi}(t)}\big]=0$ and in the third to last line the definition of the flow, Eq.~(\ref{eq:flow}).

\subsection{Proof of Eq.~\eqref{eq:Varest}}
Let us first explain why one may expect an estimate of a form similar to Eq.~\eqref{eq:Varest}. We are considering for a fixed quantum state $\ket{\Psi}$ the operator $\hat{D}_{\Psi}$ given by Eq.~(\ref{eq:Dhatf}) (the time-dependence of those quantities is not relevant and we suppress it). It is crucial that the operator $\hat{D}_{\Psi}$ has an integer-valued spectrum following from the fact that the quantum marginal constraints (\ref{eq:gpc}) have integer coefficients. Moreover, we have $\bra{\Psi}\hat{D}_{\Psi}\ket{\Psi}=D(\vec{\lambda})$. Let us now assume that $D(\vec{\lambda})$ has a finite distance to the next integer, i.e.~a finite distance to the closest eigenvalue of $\hat{D}_{\Psi}$. Consequently, the corresponding quantum state $\ket{\Psi}$ needs to have weight on at least one eigenstate with an eigenvalue larger than $D(\vec{\lambda})$ and on one with an eigenvalue smaller than $D(\vec{\lambda})$. Hence, the variance of $\hat{D}_{\Psi}$ cannot vanish (recall that the variance of an operator vanishes if and only if the state lies in an eigenspace of that operator).
Therefore, the quantity $D(\vec{\lambda})$ (by assuming it to be smaller than $1$) should provide a lower bound on $\mbox{Var}_{\Psi}\!\hat{D}_{\Psi}$. Indeed, the following lemma on random variables which we are going to prove below will establish such a relation (namely Eq.~\eqref{eq:Varest}):

\begin{lem}\label{lem:VarExp}
Let $X$ be a real-valued random variable, $a,b\in\mathbb{R}$ with $\prob(X\in(a,b))=0$, $\mu=\mean X$ the expectation value of $X$ and $\sigma^2=\mean X^2-\mu^2$ its variance such that $a<\mu<b$ . Then
\begin{equation*}
\sigma^2\ge(\mu-a)(b-\mu)
\end{equation*}
\end{lem}
\begin{proof}
By assumption, $\prob(|X-\frac{a+b}{2}|\ge\frac{b-a}{2})=1$, therefore
\begin{equation*}
\begin{split}
\mean X^2-\mu^2
 & = \mean\left[\left(X-\frac{a+b}{2}\right)^2-\left(\frac{a+b}{2}\right)^2+(a+b)X\right]-\mu^2  \\
 & \ge \left(\frac{b-a}{2}\right)^2-\left(\frac{a+b}{2}\right)^2+(a+b)\mu-\mu^2  \\
 & = (\mu-a)(b-\mu).
\end{split}
\end{equation*}
\end{proof}
We apply Lemma \ref{lem:VarExp} to the random variable $X$ corresponding to $(\hat{D}_{\Psi}, \ket{\Psi})$. This is the random variable which attains only values $X \in \mbox{spec}(\hat{D}_{\Psi})$ and those with probabilities $\bra{\Psi} \hat{P}_{\Delta}\ket{\Psi}$, where $\hat{P}_{\Delta}$ is the projector on the eigenspace of $\hat{D}_{\Psi}$ with eigenvalue $\Delta$. By choosing $a=0$ and $b=1$ estimate \eqref{eq:Varest} follows (where $\mu \equiv \bra{\Psi}\hat{D}_{\Psi}\ket{\Psi} = D(\vec{\lambda})$).

%

\subsection{Proof of Eq.~\eqref{eq:Psich}}
To derive Eq.~\eqref{eq:Psich} we apply the (second) fundamental theorem of calculus to $d\ket{\Psi(t)}/dt \equiv \ket{\dot{\Psi}(t)}$ to express for $0\leq t_1 \leq t_2$ the difference $\ket{\Psi(t_2)}-\ket{\Psi(t_1)}$ as an integral,
\begin{equation}
\begin{split}
\big\|\Psi(t_2)-\Psi(t_1)\big\|
 & = \Big\| \int_{t_1}^{t_2}\dot{\Psi}(t)dt \,\Big\| \\
 & \le \int_{t_1}^{t_2}\big\|\dot{\Psi}(t)\big\|dt  \\
 & = \int_{t_1}^{t_2}\sqrt{ \mbox{Var}_{\Psi(t)}\hat{D}_{\Psi(t)}}\,dt \\
 & = \frac{1}{\sqrt{2}}\int_{t_1}^{t_2}\sqrt{-\frac{d}{dt}D(\vec{\lambda}(t))}\,dt \\
 & = \frac{1}{\sqrt{2}}\int_{t_1}^{t_2} \frac{-\frac{d}{dt}D(\vec{\lambda}(t))}{\sqrt{-\frac{d}{dt}D(\vec{\lambda}(t))}}dt \\
 & \le \frac{1}{\sqrt{2}}\int_{t_1}^{t_2}\frac{-\frac{d}{dt}D(\vec{\lambda}(t))}{\sqrt{D(\vec{\lambda}(t))}}dt \\
 &= -\frac{1}{\sqrt{2}}\int_{D(\vec{\lambda}(t_1))}^{D(\vec{\lambda}(t_2))}\frac{dD}{\sqrt{D}} \\
 &= \sqrt{2D(\vec{\lambda}(t_1))}-\sqrt{2D(\vec{\lambda}(t_2))} \\
 & \le \sqrt{2D(\vec{\lambda}(t_1))}\,.
\end{split}
\end{equation}
In the third line we have used Eq.~(\ref{eq:NPsich}), in line four Eq.~(\ref{eq:Dchange}) and in the fourth to last line estimate (\ref{eq:DchvsD}).

\section{Brute-Force proof of the stability for the setting $(3,6)$}\label{app:SmallSet}
The generalized Pauli constraints for the Borland-Dennis setting, $(N,d)=(3,6)$, read \cite{Borl1972}
\begin{eqnarray}
&&\lambda_1\geq\lambda_2\geq \ldots\geq \lambda_6\geq 0 \label{BDordering} \\
&&\lambda_1+\lambda_6=\lambda_2+\lambda_5=\lambda_3+\lambda_4=1 \label{BDGBCa} \\
&& D(\vec{\lambda})\equiv2-(\lambda_1+\lambda_2+\lambda_4)\geq 0 \label{BDGPCb}\,.
\end{eqnarray}
Since the constraints (\ref{BDGBCa}) take the form of equalities they imply universal structural simplifications for
any state $\ket{\Psi}\in \wedge^3[\mathcal{H}_1^{(6)}]$. $\ket{\Psi}$'s most general form is namely given by
\begin{equation}\label{8SD}
|\Psi\rangle =\alpha |1,2,3\rangle+ \beta |1,2,4\rangle+ \gamma |1,3,5\rangle+ \delta |2,3,6\rangle +\nu |1,4,5\rangle+\mu |2,4,6\rangle+ \xi |3,5,6\rangle+\zeta |4,5,6\rangle \,,
\end{equation}
where the states $|k\rangle,  k=1,2,\ldots,6$, are the eigenstates of $\rho_1$. Hence, (\ref{8SD}) is a self-consistent expansion.

In case of saturation of the generalized Pauli constraint (\ref{BDGPCb}) further simplifications follow according to the selection rule \eqref{eq:Dhat} and \eqref{eq:Psi}, respectively (with the $\hat{D}_{\Psi}$-operator given by Eq.~(\ref{eq:Dhatf})),
\begin{equation}\label{BDPsiPin}
|\Psi\rangle = \alpha |1,2,3\rangle + \nu |1,4,5\rangle+\mu |2,4,6\rangle \,.
\end{equation}

We now prove that every quantum state with $\vec{\lambda}$ close to the polytope facet corresponding to saturation of (\ref{BDGPCb}) has approximately the form (\ref{BDPsiPin}). Actually, we prove even a stronger statement. We show that every quantum state can be written as
\begin{equation}\label{BDrelaxed}
|\Psi\rangle = U_{3,4}\big[\alpha |1,2,3\rangle + \nu |1,4,5\rangle+\mu |2,4,6\rangle\big]+\ket{\Psi_R} \,
\end{equation}
where $\ket{i}$ are the eigenstates of the one-particle reduced density matrix of $\ket{\Psi}$, $U_{34}$ an appropriate unitary transformation `rotating' in the subspace $\ket{3},\ket{4}$, and
\begin{equation}\label{BDerror}
  \|\Psi_R\| \leq 2 \frac{D(\vec{\lambda})}{1-D(\vec{\lambda})}\leq 4\, D(\vec{\lambda})\,.
\end{equation}

To prove this statement we will use some results already derived in Ref.~\cite{CSQuasipinning}.
The eight coefficients $\alpha,\ldots,\zeta$ obey further restrictions (`self-consistency conditions'), guaranteeing that the eigenvalues of $\rho_1$ are decreasingly ordered and that $\rho_1$ is diagonal with respect to its eigenstates $|k\rangle,  k=1,2,\ldots,6$. Consequently, we have
\begin{eqnarray}
\lambda_4 &=& |\beta|^2+|\mu|^2+|\nu|^2+|\zeta|^2   \nonumber \\
\lambda_5 &=& |\gamma|^2+|\nu|^2+|\xi|^2+|\zeta|^2  \nonumber \\
\lambda_6 &=& |\delta|^2+|\mu|^2+|\xi|^2+|\zeta|^2\,
\end{eqnarray}
and the largest three eigenvalues follow from Eq.~(\ref{BDGBCa}).
In the following, we choose $\lambda_4,\lambda_5$ and $\lambda_6$ as the free
variables.

The following two theorems proven in Ref.~\cite{CSQuasipinning} will be needed.
\begin{thm}\label{lem:xizeta}
For $|\Psi\rangle\in \wedge^3[\mathcal{H}_1^{(6)}]$ expanded according to Eq.~(\ref{8SD}) one has
\begin{equation}\label{xizetabound}
|\xi|^2+|\zeta|^2 \leq D(\vec{\lambda})\,.
\end{equation}
\end{thm}
and
\begin{thm}\label{lem:unstablebound}
For $|\Psi\rangle\in \wedge^3[\mathcal{H}_1^{(6)}]$ expanded according to Eq.~(\ref{8SD}) one finds
\begin{equation}\label{unstablebound}
\|\hat{P}_{\pi_{3,4} D}^{(\Psi)} \Psi\|_{L^2}^2=|\beta|^2+|\gamma|^2+|\delta|^2  \leq \frac{D(\vec{\lambda})}{\lambda_3-\lambda_4}+3 D(\vec{\lambda})\,.
\end{equation}
Here, $\pi_{3,4}$ denotes the swapping of the third and fourth entry of the vector $\vec{\lambda}\in \mathcal{P}\subset \R^6$ and $\hat{P}_D^{(\Psi)}$ is the projector onto the zero-eigenspace of the operator $\hat{D}_{\Psi}$ (\ref{eq:Dhatf}).
\end{thm}

The idea is now to first exploit the diagonality of $\rho_1$ with respect to its eigenstates, i.e.
\begin{eqnarray}
  0 &=& \bra{1}\rho_1\ket{6} \,\,\,\,\,= \alpha \delta^\ast +\beta \mu^\ast +\nu \zeta^\ast +\gamma \xi^\ast \label{off16}\\
  0 &=& -\bra{2}\rho_1\ket{5}= \alpha \gamma^\ast +\beta \nu^\ast +\delta \xi^\ast +\mu \zeta^\ast \label{off25}\\
  0 &=& \bra{3}\rho_1\ket{4}\,\,\,\,\,= \alpha \beta^\ast +\gamma \nu^\ast +\delta \mu^\ast +\xi \zeta^\ast \label{off34}\,.
\end{eqnarray}
For an approximate saturation of the generalized Pauli constraint $D$ we can (approximately) neglect $\xi, \zeta$ according to Theorem \ref{lem:xizeta}. Then, in a second step one can realize an orbital rotation
\begin{equation}
\ket{3}\rightarrow \ket{\tilde{3}}\,,\quad \ket{4}\rightarrow \ket{\tilde{4}}\,,\quad \ket{i}\rightarrow \ket{\tilde{i}}=\ket{i}\,,i=1,2,5,6\,,
\end{equation}
such that the state $\ket{\Psi}$ has indeed the structure (\ref{BDPsiPin}) (up to a small error of the order $D(\vec{\lambda})$).
That such a unitary transformation exists follows from Eqs.~(\ref{off16}), (\ref{off25}). Now we implement such a transformation for the case of $\xi, \zeta$ finite.

Since any quantum state $\ket{\Psi}$ always carries some significant weight on the $\alpha$ and/or $\beta$-configuration we define
\begin{equation}
\ket{\tilde{3}}\equiv \frac{\alpha\ket{3}+\beta\ket{4}}{\sqrt{|\alpha|^2+|\beta|^2}}\,,\quad \ket{\tilde{4}}\equiv \frac{-\beta^\ast\ket{3}+\alpha^\ast\ket{4}}{\sqrt{|\alpha|^2+|\beta|^2}}\,.
\end{equation}
Alternatively, this means
\begin{equation}
\ket{3}\equiv \frac{\alpha^\ast\ket{\tilde{3}}+\beta\ket{\tilde{4}}}{\sqrt{|\alpha|^2+|\beta|^2}}\,,\quad \ket{4}\equiv \frac{\beta^\ast\ket{\tilde{3}}-\alpha\ket{\tilde{4}}}{\sqrt{|\alpha|^2+|\beta|^2}}\,.
\end{equation}
On the level of $\ket{\Psi}$ the unitary transformation leads to
\begin{equation}\label{8SDU}
U_{3,4}|\Psi\rangle = \tilde{\alpha} |1,2,3\rangle+ \tilde{\beta} |1,2,4\rangle+ \tilde{\gamma} |1,3,5\rangle + \tilde{\delta} |2,3,6\rangle +\tilde{\nu} |1,4,5\rangle+\tilde{\mu} |2,4,6\rangle + \tilde{\xi} |3,5,6\rangle+\tilde{\zeta} |4,5,6\rangle \,.
\end{equation}
The hope is now that this transformation will not only imply $\tilde{\beta}=0$, according to construction, but also
$\tilde{\gamma},\tilde{\delta}\approx0$. Notice that whenever $\xi,\zeta\approx0$ we also have $\tilde{\xi},\tilde{\zeta}\approx0$.
In the following we will confirm that this is the case whenever $D(\vec{\lambda})$ is sufficiently small. First, we calculate
\begin{eqnarray}
 && \tilde{\gamma}=\frac{\alpha^\ast \gamma+\beta^\ast \nu}{\sqrt{|\alpha|^2+|\beta|^2}}\,,\quad \tilde{\nu}=\frac{\beta \gamma-\alpha \nu}{\sqrt{|\alpha|^2+|\beta|^2}} \,,\quad
  \tilde{\delta}=\frac{\alpha^\ast \delta+\beta^\ast \mu}{\sqrt{|\alpha|^2+|\beta|^2}} \nonumber \\
 && \tilde{\mu}=\frac{\beta \delta-\alpha \mu}{\sqrt{|\alpha|^2+|\beta|^2}} \,,\quad
  \tilde{\xi}=\frac{\alpha^\ast \xi+\beta^\ast \zeta}{\sqrt{|\alpha|^2+|\beta|^2}}\,,\quad \tilde{\zeta}=\frac{\beta \xi-\alpha \zeta}{\sqrt{|\alpha|^2+|\beta|^2}}\,.
\end{eqnarray}
Particularly, combining this with Eqs.~(\ref{off16}), (\ref{off25}) we find
\begin{equation}
  \tilde{\gamma} =- \frac{\delta^\ast \xi +\mu^\ast \zeta}{\sqrt{|\alpha|^2+|\beta|^2}} \,, \qquad
  \tilde{\delta} =- \frac{\nu^\ast \zeta +\gamma^\ast \xi}{\sqrt{|\alpha|^2+|\beta|^2}}\,.
\end{equation}
This leads to
\begin{eqnarray}
(|\alpha|^2+|\beta|^2) \big[|\tilde{\beta}|^2+|\tilde{\gamma}|^2+|\tilde{\delta}|^2+|\tilde{\xi}|^2+|\tilde{\zeta}|^2\big] &=&
|\delta \xi^\ast +\mu \zeta^\ast|^2+ |\gamma \xi^\ast +\nu \zeta^\ast|^2+ |\alpha \xi^\ast +\beta \zeta^\ast|^2 \, |\beta \xi^\ast -\alpha \zeta^\ast|^2 \nonumber \\
&\leq& (|\alpha|^2+|\beta|^2+|\gamma|^2+|\delta|^2+|\mu|^2+|\nu|^2) \, (|\xi|^2+|\zeta|^2)
\end{eqnarray}
and thus
\begin{equation}
|\tilde{\beta}|^2+|\tilde{\gamma}|^2+|\tilde{\delta}|^2+|\tilde{\xi}|^2+|\tilde{\zeta}|^2
   \leq \frac{D(\vec{\lambda})}{|\alpha|^2+|\beta|^2}\,.
\end{equation}
For the last estimate, we have used the normalization of the quantum state and Theorem \ref{lem:xizeta}.

Eventually, we still need to estimate $|\alpha|^2+|\beta|^2$.
By using $Q=P_{3,4}D$, where $P_{3,4}$ swaps $\lambda_3$ and $\lambda_4$ we obtain
\begin{eqnarray}
  2 D(\vec{\lambda}) &\geq & D(\vec{\lambda})+Q(\vec{\lambda}) \nonumber \\
  &=& -|\beta|^2 +|\gamma|^2+ |\delta|^2+ 2|\xi|^2 +|\zeta|^2 \nonumber \\
  &=& -|\alpha|^2 +|\nu|^2+ |\mu|^2+ 2|\zeta|^2 +|\xi|^2 \nonumber \\
  &=& -2 (|\alpha|^2+|\beta|^2) +1 + 2 (|\xi|^2+|\zeta|^2)\,.
\end{eqnarray}
By using Theorem \ref{lem:xizeta} we find
\begin{equation}
  |\alpha|^2+|\beta|^2 \geq \frac{1}{2} +|\xi|^2+|\zeta|^2-D(\vec{\lambda}) \geq \frac{1}{2} -D(\vec{\lambda})
\end{equation}
and then eventually
\begin{equation}
 1-(|\tilde{\alpha}|^2+|\tilde{\nu}|^2+|\tilde{\mu}|^2) \leq 2 \frac{D(\vec{\lambda})}{1-D(\vec{\lambda})}\,,
\end{equation}
which completes the proof.

\section{Relating energy, $N$-fermion and one-fermion picture}\label{app:EvsN}

\subsection{Relating energy and $N$-particle picture.}
The weight of an arbitrary quantum state outside the ground state space of a Hamiltonian can be estimated by the energy expectation value of that state relative to the ground state energy:
\begin{lem}\label{lem:EvsN}
Let $\hat{H}$ be a Hamiltonian on a finite Hilbert space $\mathcal{H}$, where $\hat{\pi}_{E_0}$ denotes the projector on the ground state subspace, $E_{ex}^{(-)}$ the energy of the first excited state and $E_{ex}^{(+)}$ the one of the highest excited state. Then for any $\ket{\tilde{\Psi}}\in\mathcal{H}$ with energy expectation value $\tilde{E}\equiv \bra{\tilde{\Psi}}\hat{H}\ket{\tilde{\Psi}}$ one has
\begin{equation}\label{eq:EvsN}
\frac{\tilde{E}-E_0}{E_{ex}^{(+)}-E_0} \leq 1-\|\hat{\pi}_{E_0} \tilde{\Psi}\|^2\leq \frac{\tilde{E}-E_0}{E_{ex}^{(-)}-E_0}\,,
\end{equation}
 Eq.~(\ref{eq:EvsN}) is a universal (i.e.~for all $\hat{H}$) relation between the energy picture and the $N$-particle picture.
\end{lem}
\begin{proof}
By using the spectral decomposition of the Hamiltonian, $\hat{H}=\sum_{E'} E' \hat{\pi}_{E'}$ we obtain
\begin{eqnarray}\label{eq:estED0}
\tilde{E} &\equiv& \bra{\tilde{\Psi}}\hat{H}\ket{\tilde{\Psi}} \nonumber \\
&=& \sum_{E'} E' \bra{\tilde{\Psi}}\hat{\pi}_{E'}\ket{\tilde{\Psi}} \nonumber \\
&=& E_0 \bra{\tilde{\Psi}}\hat{\pi}_{E_0}\ket{\tilde{\Psi}} +\sum_{E'>E_0} E' \bra{\tilde{\Psi}}\hat{\pi}_{E'}\ket{\tilde{\Psi}} \nonumber \\
&\leq &  E_0 \bra{\tilde{\Psi}}\hat{\pi}_{E_0}\ket{\tilde{\Psi}}+ E_{ex}^{(+)} \bra{\tilde{\Psi}}(\mathbb{1}-\hat{\pi}_{E_0})\ket{\tilde{\Psi}} \nonumber \\
&=& E_0+(E_{ex}^{(+)}-E_0)\, \big(1-\|\hat{\pi}_{E_0} \tilde{\Psi}\|^2\big)\,.
\end{eqnarray}
In the second to last line we have estimated every excited energy from above by the maximal excited energy $E_{ex}^{(+)}$ and the lower bound in (\ref{eq:estED0}) follows immediately. Repeating the same derivation but by estimating in the second to last line every excited energy from below by the minimal excited energy $E_{ex}^{(-)}$ yields then the upper bound in (\ref{eq:estED0}).
\end{proof}

\subsection{Proof of estimate \eqref{eq:estimate1}}
Let $\ket{\Psi_D}\in \mathcal{M}_D$ be the variational minimizer of the energy expectation value. By denoting the reference basis by $\mathcal{B}_1$ (also obtained from the energy minimization), $\ket{\Psi_D}$ lies within the zero-eigenspace of $\hat{D}_{\mathcal{B}_1}$.
By using the projection operator $\hat{P}_D^{(\mathcal{B}_1)}$, projecting onto the zero-eigenspace of $\hat{D}_{\mathcal{B}_1}$, we define
\begin{equation}
\ket{\tilde{\Psi}_D} \equiv \frac{\hat{P}_D^{(\mathcal{B}_1)} \ket{\Psi_0}}{\big\|\hat{P}_D^{(\mathcal{B}_1)} \Psi_0\big\|}\,,
\end{equation}
where $\|\cdot\|$ denotes the $L^2$-norm and $\langle \tilde{\Psi}_D\ket{\tilde{\Psi}_D}=1$. Then, by using the spectral decomposition of the Hamiltonian, $\hat{H}=\sum_{E'} E' \hat{\pi}_{E'}$ and by assuming that the ground state is unique, $\hat{\pi}_{E_0}=\ket{\Psi_0}\!\bra{\Psi_0}$ , we obtain
\begin{eqnarray}\label{eq:estED1}
E_D &\equiv& \bra{\Psi_D}\hat{H}\ket{\Psi_D} \nonumber \\
&=& \min_{\small\begin{array}{c}
\ket{\Phi}\in \wedge^N[\mathcal{H}_1^{(d)}]\\
\bra{\Phi}\Phi\rangle =1\end{array}}
\bra{\Phi}\hat{P}_D^{(\mathcal{B}_1)}\hat{H}\hat{P}_D^{(\mathcal{B}_1)}\ket{\Phi} \nonumber \\
&\leq& \bra{\tilde{\Psi}_D}\hat{P}_D^{(\mathcal{B}_1)}\hat{H}\hat{P}_D^{(\mathcal{B}_1)}\ket{\tilde{\Psi}_D} \nonumber \\
&=& \bra{\tilde{\Psi}_D}\hat{H}\ket{\tilde{\Psi}_D} \nonumber \\
&=& \sum_{E'} E' \bra{\tilde{\Psi}_D}\hat{\pi}_{E'}\ket{\tilde{\Psi}_D} \nonumber \\
&\leq& E_0 \bra{\tilde{\Psi}_D}\hat{\pi}_{E_0}\ket{\tilde{\Psi}_D} +\sum_{E'>E_0} E_{ex}^{(+)} \bra{\tilde{\Psi}_D}\hat{\pi}_{E'}\ket{\tilde{\Psi}_D} \nonumber \\
&=&  E_0 \bra{\tilde{\Psi}_D}\hat{\pi}_{E_0}\ket{\tilde{\Psi}_D}+ E_{ex}^{(+)} \bra{\tilde{\Psi}_D}(\mathbb{1}-\hat{\pi}_{E_0})\ket{\tilde{\Psi}_D} \nonumber \\
&=& E_0+(E_{ex}^{(+)}-E_0)\, \big(1-|\bra{\Psi_0}\tilde{\Psi}_D\rangle|^2\big)\,.
\end{eqnarray}
In the second line, for the specific $\mathcal{B}_1$, we have used the fact that $\ket{\Psi_D}$ is the variational minimizer of the energy expectation value within the zero-eigenspace of $\hat{D}_{\mathcal{B}_1}$ (onto which $\hat{P}_D^{(\mathcal{B}_1)}$ projects).
In the third to last line we have bounded every excited energy from above by the maximal excited energy
$E_{ex}^{(+)}$. Estimate (\ref{eq:estED1}) yields the result (\ref{eq:estimate1}):
\begin{eqnarray}\label{eq:17}
\Delta E_D \equiv E_D-E_0 &\leq&  (E_{ex}^{(+)}-E_0) \,\big(1-|\bra{\Psi_0}\tilde{\Psi}_D\rangle|^2\big)\nonumber \\
&=& (E_{ex}^{(+)}-E_0)\,\big(1-\|\hat{P}_D^{(\mathcal{B}_1)} \Psi_0\|^2\big) \nonumber \\
&\leq & 2 (E_{ex}^{(+)}-E_0)\,D(\vec{\lambda}_0)\,,
\end{eqnarray}
where we have used the main result, Eq.~(\ref{eq:WEst}), in the last line.

\subsection{Proof of estimate \eqref{eq:estimate2}}
To prove estimate (\ref{eq:estimate2}) we also need to relate the energy picture with the $N$-particle picture for the Hartree-Fock ansatz. Estimate (\ref{eq:estimate1}) holds of course for any variational ansatz based on extremal local information and therefore in particular for the Hartree-Fock ansatz. Yet, we need for Eq.~(\ref{eq:estimate2}) a reversed version of (\ref{eq:estimate1}). By following closely (\ref{eq:estED1}) we find
\begin{eqnarray}\label{eq:estES1}
E_{HF} &\equiv& \bra{\Psi_{HF}}\hat{H}\ket{\Psi_{HF}} \nonumber \\
&=& \sum_{E'} E' \bra{\Psi_{HF}}\hat{\pi}_{E'}\ket{\Psi_{HF}} \nonumber \\
&\geq & E_0 \bra{\Psi_{HF}}\hat{\pi}_{E_0}\ket{\Psi_{HF}} +\sum_{E'>E_0} E_{ex}^{(-)} \bra{\Psi_{HF}}\hat{\pi}_{E'}\ket{\Psi_{HF}} \nonumber \\
&=&  E_0 \,\bra{\Psi_{HF}}\hat{\pi}_{E_0}\ket{\Psi_{HF}}+ E_{ex}^{(-)} \bra{\Psi_{HF}}(\mathbb{1}-\hat{\pi}_{E_0})\ket{\Psi_{HF}} \nonumber \\
&=& E_0+(E_{ex}^{(-)}-E_0)\, \big(1-|\bra{\Psi_0}\Psi_{HF}\rangle|^2\big)\,.
\end{eqnarray}
In the third line we have bounded every excited energy from below by the minimal excited energy
$E_{ex}^{(-)}$ and in the last line we have used that the ground state is unique, i.e.~$\hat{\pi}_{E_0}\equiv \ket{\Psi_0}\!\bra{\Psi_0}$.
Estimate (\ref{eq:estES1}) then leads to
\begin{equation}\label{eq:estES2}
E_{corr} \equiv E_{HF}-E_0 \geq  (E_{ex}^{(-)}-E_0) \,\big(1-|\bra{\Psi_0}\Psi_{HF}\rangle|^2\big)\,.
\end{equation}

Now, to connect the $N$-particle picture to the one-particle picture we will need the following lemma
\begin{lem}\label{lem:HF}
For $\ket{\Psi}\in \wedge^N[\mathcal{H}_1^{(d)}]$ and any orthonormal basis $\{\ket{i'}\}_{i=1}^d$ for the one-particle Hilbert space $\mathcal{H}_1^{(d)}$ we have
\begin{equation}
\frac{S(\vec{\lambda})}{N} \leq 1-\big|\bra{1',2',\ldots,N'} \Psi\rangle\big|^2\,,
\end{equation}
where $\vec{\lambda}=(\lambda^{(i)})_{i=1}^d$ is the non-increasingly ordered spectrum of the one-particle reduced density matrix $\rho_1$ (\ref{eq:1rdm}) of $\ket{\Psi}$ and $\ket{1',\ldots,N'}$ denotes the Slater determinant built up from the states $\ket{i'}, i'=1,\ldots,N$. $S(\vec{\lambda})\equiv \sum_{i=1}^N (1-\lambda^{(i)})+\sum_{j=N+1}^d \lambda^{(j)}$ is the $l^1$-distance of $\vec{\lambda}$ to the Hartree-Fock point $\vec{\lambda}_{HF}\equiv (1,\ldots,1,0,\ldots,0)$.
\end{lem}
\begin{proof}
The proof of Lemma \ref{lem:HF} is elementary. First, we consider the operator,
\begin{equation}
\hat{S}'\equiv \sum_{i=1}^N (\mathbb{1}-\hat{n}'_i)+ \sum_{j=N+1}^d \hat{n}'_j\,,
\end{equation}
where $\hat{n}'_i$ denotes the particle number operator for the state $\ket{i'}$ (acting on the $N$-fermion Hilbert space) and we define the occupancies $\lambda'_i\equiv \bra{\Psi}\hat{n}'_i\ket{\Psi}$. Then, by using the spectral decomposition $\hat{S}'=\sum_{s=0}^N s \hat{P}'_s$ we find
\begin{eqnarray}\label{eq:estimateHF}
S(\vec{\lambda}')&\equiv& \bra{\Psi}\hat{S}'\ket{\Psi} \nonumber \\
&=& \sum_{s=1}^N s\,\|\hat{P}'_s\Psi\|^2 \nonumber \\
&\leq& \sum_{s=1}^N N \,\|\hat{P}'_s\Psi\|^2 \nonumber \\
&=& N \,\big(1-\|\hat{P}'_0\Psi\|^2\big)\nonumber \\
&\leq& N \,\Big(1-\big|\bra{1',2',\ldots,N'} \Psi\rangle\big|^2\Big)\,.
\end{eqnarray}
Since $\vec{\lambda}$, the vector of non-increasingly ordered eigenvalues of $\rho_1$, majorizes any vector of occupation numbers, as, e.g., $\vec{\lambda}'$, we find
$S(\vec{\lambda})\leq S(\vec{\lambda}')$, which together with estimate (\ref{eq:estimateHF}) completes the proof of Lemma \ref{lem:HF}.
\end{proof}

By choosing in Lemma \ref{lem:HF} the Hartree-Fock molecular orbitals $\{\ket{\chi_j}\}_{j=1}^d$ as reference basis $\{\ket{i'}\}_{j=1}^d$ we obtain
\begin{equation}
\frac{S(\vec{\lambda}_0)}{N} \leq 1-\big|\bra{\chi_1,\ldots,\ldots,\chi_N} \Psi_0\rangle\big|^2 \equiv 1-\big|\bra{\Psi_{HF}} \Psi_0\rangle\big|^2 \,.
\end{equation}
This together with Eq.~(\ref{eq:estES2}) yields
\begin{equation}\label{eq:17HF}
E_{corr} \geq    \frac{E_{ex}^{(-)}-E_0}{N}\,S(\vec{\lambda}_0)\,,
\end{equation}
Estimate (\ref{eq:17HF}) in combination with Eq.~(\ref{eq:estimate1}) leads to the final result, Eq.~(\ref{eq:estimate2}),
\begin{equation}
\frac{\Delta E_D}{E_{corr}}
 \leq K\, \frac{D(\vec{\lambda}_0)}{S(\vec{\lambda}_0)}\,,
\end{equation}
where $E_{corr}\equiv E_{HF}-E_0$ is the correlation energy and $K\equiv 2 N (E_{ex}^{(+)}-E_0)/(E_{ex}^{(-)}-E_0)$.
\end{document}